\documentclass[preprint, 12pt]{elsarticle}

\usepackage{amssymb}

\usepackage{amssymb}

\usepackage{amsmath}

\usepackage{makecell}

\usepackage{bm} 

\usepackage{color}

\usepackage{enumerate}
\usepackage{amsthm}
\usepackage{xpatch}
\xpatchcmd{\proof}{\itshape}{\prooflabelfont}{}{}
\newcommand{\prooflabelfont}{\bfseries}

\graphicspath{{figure/}}                                 

\newtheorem{example}{Example}
\newtheorem{definition}{Definition}
\newtheorem{lemma}{Lemma}
\newtheorem{theorem}{Theorem}
\newtheorem{proposition}{Proposition}
\newtheorem{remark}{Remark}
\newtheorem{corollary}{Corollary}

\usepackage{hyperref}
\hypersetup{
	colorlinks=true,
	linkcolor=green,
	filecolor=green,      
	urlcolor=green,
	citecolor=green,
}

\DeclareUnicodeCharacter{2212}{\ensuremath{-}}

\DeclareMathOperator{\D}{d_{min}}

\DeclareMathOperator{\Wt}{wt}

\DeclareMathOperator{\Supp}{supp}

\journal{Finite Fields and Their Applications}

\begin{document}

\begin{frontmatter}



\title{New constructions of optimal $(r, \delta)$-LRCs via good polynomials}

\author{Yuan Gao}
\ead{gaoyuan862023@163.com}
\author{Siman Yang\corref{cor}}
\ead{smyang@math.ecnu.edu.cn}
\address{School of Mathematical Sciences,  Key Laboratory of MEA (Ministry of Education) \emph{\&} Shanghai Key Laboratory of PMMP,  East China Normal University, Shanghai 200241, China}
\cortext[cor]{Corresponding author}
%

%

\begin{abstract}
Locally repairable codes (LRCs) are a class of erasure codes that are widely used in distributed storage systems, which allow for efficient recovery of data in the case of node failures or data loss. 
In 2014, Tamo and Barg introduced Reed-Solomon-like (RS-like) Singleton-optimal $(r,\delta)$-LRCs based on polynomial evaluation. These constructions rely on the existence of so-called good polynomial that is constant on each of some pairwise disjoint subsets of $\mathbb{F}_q$.
In this paper, we extend the aforementioned constructions of RS-like LRCs and proposed new constructions of $(r,\delta)$-LRCs whose code length can be larger. These new $(r,\delta)$-LRCs are all distance-optimal, namely, they attain an upper bound on the minimum distance, that will be established in this paper. This bound is sharper than the Singleton-type bound in some cases owing to the extra conditions, it coincides with the Singleton-type bound for certain cases.
Combing these constructions with known explicit good polynomials of special forms, we can get various explicit Singleton-optimal $(r,\delta)$-LRCs with new parameters, whose code lengths are all larger than that constructed by the RS-like $(r,\delta)$-LRCs introduced by Tamo and Barg.
Note that the code length of classical RS codes and RS-like LRCs are both bounded by the field size. We explicitly construct the Singleton-optimal $(r,\delta)$-LRCs with length $n=q-1+\delta$ for any positive integers $r,\delta\geq 2$ and $(r+\delta-1)\mid (q-1)$. When $\delta$ is proportional to $q$, it is asymptotically longer than that constructed via elliptic curves whose length is at most $q+2\sqrt{q}$. Besides, it allows more flexibility on the values of $r$ and $\delta$.
\end{abstract}

%


\end{frontmatter}


\section{Introduction}\label{sec:1Intro}
In the era of big data, distributed storage systems are crucial for the storage of huge amount of data. 
Compared with centralized storage, distributed storage has obvious advantages in improving storage reliability and reducing storage load. In distributed storage systems, the raw data is encoded and stored in different
data storage devices. Locally repairable codes is a class of erasure codes that can reduce the repair cost in distributed storage systems. It was first introduced by Gopalan \textit{et al.} \cite{gopalan2012locality} in 2012.

An $[n,k,d]_q$ linear code $\mathcal{C}$ is called an LRC with locality $r$ or an $r$-LRC, if for any $i\in\{1,2,\dots,n\}$, there exists a subset $R_i\subseteq \{1,2,\dots,n\}$ containing $i$, such that $|R_i|\leq r+1$ and for any codeword $(c_1,c_2,\dots,c_n)\in \mathcal{C}$, $c_i$ can be determined by the other $(|R_i|-1)$ code symbols $c_i,i\in R_i\backslash \{i\}$. For classical codes, there exists a trade-off $k+d\leq n+1$ between dimension $k$ and minimum distance $d$ which is known as the Singleton bound. For $r$-LRCs, 
the Singleton-type bound \eqref{S-L, 1} was proposed in \cite{gopalan2012locality}.
\begin{eqnarray}\label{S-L, 1}
	d\leq n-k-\lceil\frac{k}{r}\rceil+2.
\end{eqnarray}
To
address the situation of multiple device failures, the notion of \emph{locality $r$} was generalized in \cite{prakash2012optimal}. An $[n,k,d]_q$ linear code $\mathcal{C}$ is called an $(r,\delta)$-LRC, where $\delta \geq2$, if for any $i\in \{1,2,\dots,n\}$, there exists a subset $R_i\subseteq \{1,2,\dots,n\}$ containing $i$, such that $|R_i|\leq r+\delta-1$ and any erasure of $(\delta-1)$ code symbols at $R_i$ can be recovered by accessing the other $(|R_i|-\delta+1)$ code symbols at $R_i$. 
The following upper bound on the minimum distance of an $(r,\delta)$-LRC was proposed in \cite{prakash2012optimal}.
\begin{equation}\label{S-L, 2}
	d\leq n-k-(\lceil\frac{k}{r}\rceil-1)(\delta-1)+1.
\end{equation}
It is easy to see that an $(r,2)$-LRC is an $r$-LRC and the bound \eqref{S-L, 2} degenerates to the bound \eqref{S-L, 1} when $\delta=2$. An $(r,\delta)$-LRC attaining the  bound \eqref{S-L, 2} is called a Singleton-optimal $(r,\delta)$-LRCs. 

Locally repairable codes have attracted the attention of many researchers in 
recent years. There have been many constructions of LRCs attaining these bounds, e.g., \cite{tamo2014family}, \cite{kolosov2018optimal}, \cite{rajput2020optimal}, \cite{li2018optimal}, \cite{jin2019construction}, \cite{jin2019explicit},  \cite{xi2022optimal}, \cite{ma2019optimal}, \cite{chen2017constructions}, \cite{chen2019constructions}, \cite{hao2017linear}. In particular, the well-known RS-like Singleton-optimal LRCs were constructed via polynomial evaluation in \cite{tamo2014family}, whose lengths are at most $q$. It is worth noting that these constructions only require $n,k,r$ to satisfy $(r+1)\mid n$ and $1\leq \frac{k}{r}\leq \frac{n}{r+1}$. In \cite{jin2019construction}, Jin \textit{et al.} generalized the constructions in \cite{tamo2014family}, constructing Singleton-optimal LRCs with code lengths up to $q+1$ and more flexible locality $r$ by employing automorphism groups of the rational function
fields. In \cite{chen2017constructions} and \cite{chen2019constructions}, Chen \textit{et al.} constructed Singleton-optimal $(r,\delta)$-LRCs with length $n$ and dimension $k$ for any $q$ satisfying $n\mid (q+1)$, $(r+\delta-1)\mid n$ and $2\leq r\leq k$ via cyclic and constacyclic codes. These LRCs are also of great value in application since some efficient decoding algorithms can be applied to them. In \cite{li2018optimal}, by using rich algebraic structures of elliptic curves, Singleton-optimal $r$-LRCs were constructed for small localities $r=2,3,5,7,11,23$ with length $n\leq q+2\sqrt{q}$.
By shortening codes from the RS-like LRCs constructed in \cite{tamo2014family}, Kolosov \textit{et al.} \cite{kolosov2018optimal} proposed the explicit constructions of distance-optimal $r$-LRCs with length $n\equiv 2,3,\dots,r\mod r+1$. For some dimension $k$, these LRCs are Singleton-optimal.
The constructions in \cite{tamo2014family} and \cite{kolosov2018optimal} are all based on the so-called good polynomial that is constant on each of some point sets with the same cardinality (see Definition \ref{good polynomial, defn2}). The code length is limited by the good polynomials. 
\subsection{Our results and comparison}	
Motivated by the constructions of LRCs in \cite{tamo2014family} and \cite{kolosov2018optimal}, we construct some distance-optimal $(r,\delta)$-LRCs with slightly larger code length than that constructed in \cite{tamo2014family} for some good polynomials of special forms. These $(r,\delta)$-LRCs are Singleton-optimal for some dimension $k$. More explicitly, we have constructions of $(r,\delta)$-LRCs for the following two classes of good polynomials $g(x)$.
In the following statements, $g(x)\in \mathbb{F}_q[x]$ is a good polynomial with degree $(r+\delta-1)$ that takes non-zero constants on exactly $L$ pairwise disjoint $(r+\delta-1)$-subsets, in which case the code length of the $(r,\delta)$-LRCs constructed by Tamo and Barg in \cite{tamo2014family} is at most $L(r+\delta-1)$.

A. Assume that $g(x)$ has $s\geq \delta$ distinct roots, then we can employ these $s$ roots as extra evaluation points to increase the code length by $s$, to $L(r+\delta-1)+s$. The maximal code length of its explicit constructions is $q$.

B. Assume that $g(x)$ can be factored into the product of two polynomials with degrees $r-1$ and $\delta$ respectively, then the code length can be increased by $\delta$, to $L(r+\delta-1)+\delta$. The maximal code length of its explicit constructions is $q-1+\delta$.

In Proposition \ref{prop:MiniDistanceUpperBound}, we prove an upper bound on the minimum distance that is sharper than the Singleton-type \eqref{S-L, 2} by requiring the local repair groups to have particular form.
It turns out that the above $(r,\delta)$-LRCs are all distance-optimal in the sense of attaining this bound. They are Singleton-optimal if $1\leq (k\mod r)\leq ((n+r)\mod (r+\delta-1))$, where $(k\mod r)$ denote the least non-negative remainder when $k$ is divided by $r$.

Our constructions rely on the existence of good polynomials of special forms. There have been many constructions of good polynomials, e.g., \cite{tamo2014family}, \cite{micheli2019constructions}, \cite{liu2020constructions}, \cite{dukes2022optimal}, \cite{chen2021good}, \cite{chen2022function}.
In Table \ref{tab:table1}, we list the references of some known explicit constructions of good polynomials and the code lengths that can be increased by our constructions compared to that in \cite{tamo2014family}. The details will be shown in Lemma \ref{lem:ConsBlem1}, \ref{lem:ConsBlem2}, \ref{lem:ConsClem1}, \ref{lem:ConsClem2} and Corollary \ref{cor:ConsBcor1}, \ref{cor:ConsBcor2}, \ref{cor:ConsCcor1}, \ref{cor:ConsCcor2}, respectively.
\setcounter{table}{0}
\renewcommand{\thetable}{\arabic{table}}
\begin{table}[tbhp] 
	\footnotesize
	\caption{Some known good polynomials and the code length increased by our constructions}\label{tab:table1}
	\begin{center}
		\begin{tabular}{|c|c|c|c|p{2.7cm}|}\hline
			\textbf{\makecell{Constructions of\\Good polynomials}}&\textbf{\makecell{Degree}}& \textbf{\makecell{The code length \\is increased by}}&\textbf{\makecell{Code length}}                                                          \\ \hline
			
			\cite[Theorem 3.3]{tamo2014family}              &    $m|H|$       &     $|H|$  &      $q$                       \\ \hline
			
			\cite[Theorem 19]{chen2022function}           &    $m$                  & $(m+1)/2$  & $m\lfloor q/2m\rfloor+(m+1)/2$ \\ \hline   
			
			\cite[Proposition 3.2]{tamo2014family}                     &   $|H|$        &     $\delta$      &   $q+\delta-1$                        \\ \hline
			
			\cite[Theorem 3]{chen2021good}                            &   $3$          &      $2$   &  $(q+5)/2$                     \\ \hline
		\end{tabular}
	\end{center}
\end{table}

In particular, by Corollary \ref{cor:ConsBcor1}, if $q=p^s,1\leq l\textless s$, $2\leq m|(p^l-1)$ and $l\mid s$, where $p$ is a prime, then there is a $[q,k,d]_q$ LRC with locality $r=mp^l-1$ by letting $H=\mathbb{F}_{p^l}$.
When $k\equiv 1,2,\dots,p^l-1\mod r$, it is a Singleton-optimal $r$-LRC with length $n=q$. By Corollary \ref{cor:ConsCcor1}, if $r,\delta\geq 2$ and $(r+\delta-1)|(q-1)$, then there exists a $[q+\delta-1,k,d]_q$ LRC with $(r,\delta)$-locality. When $k\equiv 1\mod r$, it is a Singleton-optimal $(r,\delta)$-LRC with length $n=q+\delta-1$, which is longer than that constructed via elliptic curves in \cite{li2018optimal} when $\delta$ is proportional to $q$. In Table \ref{tab:table2}, we list some known constructions of Singleton-optimal $r$-LRCs whose minimum distance can be proportional to the code length, where we assume $\delta=2$ for simplicity.
\begin{table}[tbhp] \scriptsize
	\caption{Some known constructions of Singleton-optimal LRCs with large minimum distance}\label{tab:table2}
	\begin{center}
		\begin{tabular}{|c|c|c|c|c|p{2.7cm}|}\hline
			\textbf{Length $n$}&\textbf{Locality $r$}&\textbf{Minimum distance $d$}&\textbf{References} \\ \hline
			
			$q$     & \makecell{$(r+1)=p^l$, \mbox{where} $1\leq l\leq \log_p(q)$,\\ $p$ is a prime.}&$d\leq n$ &   \cite{tamo2014family}  \cite{jin2019construction}     \\ \hline
			
			$q$    &\makecell{$(r+1)=mp^l$, \mbox{where} $1\leq l\textless \log_p(q),$\\ $l\mid \log_p(q)$, $2\leq m|(p^l-1)$.}&  $d\equiv 2,3,\dots,p^l\mod (r+1)$ & Corollary \ref{cor:ConsBcor1} \\ \hline
			
			$q+1$     &    $(r+1)|n$      &$d\leq n$ &     \cite{jin2019construction},\cite{chen2017constructions},\cite{chen2019constructions}     \\ \hline
			
			$q+1$    &$(r+1)|(q-1)$&  $d\equiv 2\mod (r+1)$  &Corollary \ref{cor:ConsCcor1} \\ \hline
			
			$q+2\sqrt{q}$  & $r=2,3,5,7,11,23$,$(r+1)|n$ &       $(r+1)|d$      &  \cite{li2018optimal}  \\ \hline   
			
			$\leq q+2\sqrt{q}$    &$2,3,4,6,8,12,24\mid (r+1)$,$(r+1)|n$  &  $(r+1)|d$     &\cite{ma2023group}  \\ \hline
		\end{tabular}
	\end{center}
\end{table}		

\smallskip
The rest of this paper is organized as follows. In  section \ref{sec:2Pre}, we review some definitions and necessary results. In section \ref{sec:3Cons}, we give our main constructions. In section \ref{sec:4Conc}, we conclude the paper.
\section{Preliminaries}\label{sec:2Pre}
Let us recall some notations and definitions.
\begin{itemize}
	\item Let $\lfloor a\rfloor,\lceil a \rceil$ be the floor function and the ceiling function of $a$, respectively.
	\item Let $\mathbb{F}_q$ be the finite field with $q$ elements, where $q$ is a prime power. 
	\item Let $\mathbb{F}_q^n$ be the $n$-dimensional vector space over $\mathbb{F}_q$.
	\item  For any $\boldsymbol{x}=(x_1,x_2,\dots,x_n)\in \mathbb{F}_q^n$, the support set of $\boldsymbol{x}$ is defined by $\Supp(\boldsymbol{x})=\{i\in [n]|x_i\neq 0\}$ and the Hamming weight of $\boldsymbol{x}$ is defined by $\Wt(\boldsymbol{x})=|\Supp(\boldsymbol{x})|$.
	\item An $[n,k]_q$ linear code $\mathcal{C}$ is a $k$-dimensional subspace of $\mathbb{F}_q^n\quad (1\leq k\leq n)$. The minimum distance of $\mathcal{C}$ is defined to be $\D(\mathcal{C})\triangleq \min\limits_{ \boldsymbol{0}\neq \boldsymbol{c}\in \mathcal{C} }\{\Wt(\boldsymbol{c})\}$. $\mathcal{C}$ is also called an  $[n,k,\D(\mathcal{C})]_q$ linear code.
	\item For any polynomial $f(x)$ over $\mathbb{F}_q$, we use $\deg(f)$ to denote the degree of $f$. In particular, the degree of the zero polynomial is specified to be $-\infty$.
	\item For any positive integers $a$ and $b$, we use $(a\mod b)$ to denote the least non-negative remainder when $a$ is divided by $b$.
	\item For any positive integer $n$, we use $[n]$ to denote the set $\{1,2,\dots,n\}$. In particular, $[0]\triangleq \varnothing$.
	\item For any finite set $A$, we call it an $n$-set if $|A|=n$. 
	
	\item Suppose  that $\mathcal{C}$ is a linear code over $\mathbb{F}_q$ with length $n$ and $S$ is a non-empty subset of $[n]$, punctured code at the set $[n]\backslash$ is denoted by $\mathcal{C}_{S}:=\{(c_i,i\in S)|(c_1,c_2,\dots,c_n)\in \mathcal{C}\}$.
\end{itemize}
\begin{definition}[Locally repairable codes]\label{defn1}
	Suppose that $\mathcal{C}$ is an $[n,k,d]_q$ linear code, if for any $i\in [n]$, there exists a subset $R_i\subseteq [n]$, such that one of the following statement hold:
	\begin{itemize}
		\item $i\in R_i$, $|R_i|\leq r+\delta-1
		$ and $\mathcal{C}_{R_i}=\{\boldsymbol{0}\}$,
		\item 
		$i\in R_i$, $|R_i|\leq r+\delta-1
		$ and $\D(\mathcal{C}_{R_i})\geq \delta$,
	\end{itemize}
	
	then we call $\mathcal{C}$ an LRC with $(r,\delta)$-locality or an $(r,\delta)$-LRC. 
	$R_i$ is called a local repair group of $\mathcal{C}$.
\end{definition}
This definition is equivalent to that described by the locally repair ability in section \ref{sec:1Intro}.
\begin{definition}[{{\cite[Sect. \uppercase\expandafter{\romannumeral3}-A]{tamo2014family}\label{good polynomial, defn2}}}]
	For a polynomial $g(x)\in \mathbb{F}_q[x]$ of degree $m\geq 2$, if there exist $L$ pairwise disjoint $m$-subsets $A_1,A_2,\dots,A_L\subseteq \mathbb{F}_q$ such that $g(x)$ is constant on each of these sets, then we call $g(x)$ a good polynomial.
\end{definition}
If a polynomial $g(x)$ is constant on the set $A$, then we use $g(A)$ to denote this constant for simplicity. Chen \textit{et al.} \cite{chen2021good} introduced some parameters that measure how "good" a polynomial is.  
\begin{definition}[{{\cite[Sect. 2.1]{chen2021good}}}]:\label{N_f, defn}
	For a polynomial $f(x)\in \mathbb{F}_q[x]$ with $\deg(f)\geq 2$,
	
	the map $N_f : \mathbb{N}\rightarrow \mathbb{N}$ is defined by:
	$$N_f(i)=|\left\{t\in \mathbb{F}_q : f(x)+t \quad\mbox{has exactly $i$ distinct roots in }\mathbb{F}_q \right\}|,$$
	
	the map
	$G: \mathbb{F}_q[x]\rightarrow \mathbb{N}$ is defined by:
	$$G(f)=N_f(\deg(f)).$$
\end{definition}

For a good polynomial $f(x)$ of degree $m\geq 2$, the code length of the RS-like LRCs constructed in \cite{tamo2014family} and \cite{kolosov2018optimal} is upper bounded by $mG(f)$. In Remark \ref{rem:ConsBrem} (3), we will show that the code length can be increased by $s$ if $N_f(s)\geq 1$, where $m\geq 3$ and $2 \leq s\leq m-1$.
\hspace*{\fill}

Assume $2\leq r\leq k$. In this paper, we denote $k'=k+v$, where $v\in [r-1]$. 
For $i\in\{0,1,\dots,r-1\}$, we define 
$$\xi(i)=\begin{cases}
	\lfloor\frac{k'}{r}\rfloor&i\text{ }\textless \text{ }(k' \mod r)\\
	\lfloor\frac{k'}{r}\rfloor-1&i\geq (k'\mod r).\\
\end{cases}$$		
Note that 
\begin{align*}
	&r-v+\lfloor\frac{k'}{r}\rfloor (k'\mod r)+(\lfloor\frac{k'}{r}\rfloor-1)(r-(k'\mod r))\\
	&=r-v+r\lfloor\frac{k'}{r}\rfloor-r+(k'\mod r)\\
	&=k'-v\\
	&=k.
\end{align*}

Hence the vector $\boldsymbol{I}\in \mathbb{F}_q^k$ can be written as 
\begin{equation}\label{Infromation, I, form}
	\boldsymbol{I}=(I_0,I_1,\dots,I_{r-v-1};{I_{i,j},i=0,1,\dots,r-1,j\in [\xi(i)]}).
\end{equation}
This expression will be used in our code constructions in section \ref{sec:3Cons}.

%
%

First we establish an improved bound which is sharper than Singleton-type bound \eqref{S-L, 2} in certain cases at the cost of some mild conditions. It is based on the similar ideas as that in \cite[Theorem III.3]{kolosov2018optimal} and \cite[Theorem 6]{rajput2020optimal}. The main difference is that this theorem is about $(r,\delta)$-LRCs rather than $(r,2)$-LRCs.
\begin{proposition}\label{prop:MiniDistanceUpperBound}
	Let $\mathcal{C}$ be an $[n,k,d]_q$ $(r,\delta)$-LRC. Suppose that $L\geq 1$ and $\mathcal{C}$ has $L+1$  disjoint local
	repair groups $A_i(i\in [L+1])$ such that $|A_i|=r+\delta-1,i=1,\dots,L$ and
	$|A_{L+1}|=r+\delta-1-v\triangleq s$, where $1\leq v\leq r-1$. 
	
	If $r\leq k\leq Lr+s-\delta+1$, then we have \
	\begin{equation}\label{eq:MiniDistanceUpperBound}
		d\leq n-k-(\lceil\frac{k+v}{r}\rceil-1)(\delta-1)+1.
	\end{equation}
\end{proposition}
\begin{proof}
	For the $[n,k,d]_q$ linear code $\mathcal{C}$, we have \begin{equation}\label{minimum distance}
		d=n-\max\limits_{S\subseteq \{1,2,\dots,n\}}\left\{|S|\bigg||\mathcal{C}_{S}|\textless q^k\right\}.
	\end{equation}
	We will show that there exists a set $S\subseteq \{1,2,\dots,n\}$ such that $|S|\geq k-1+(\lceil\frac{k+v}{r}\rceil-1)(\delta-1) $ and $|\mathcal{C}_{S}|\textless q^k$.
	
	Note that $0\leq \frac{k-1-s+\delta-1}{r}\leq L$. Let $m=\lfloor\frac{k-1-s+\delta-1}{r}\rfloor$ and $t=\lceil\frac{k-1-s+\delta-1}{r}\rceil$, we now choose $k-1$ coordinate of code $\mathcal{C}$ by the following steps, more precisely, some subsets of $A_1,A_2,\dots,A_{L+1}$ whose union has $k-1$ elements:
	\begin{enumerate}
		\item Choose an arbitrary $(s-\delta+1)$-subsets $B_{L+1}$ of $A_{L+1}$.
		\item Choose $t$ arbitrary subsets $B_1\subseteq A_1,B_2\subseteq A_2,\dots,B_t\subseteq A_t$ satisfying 
		\begin{itemize}
			\item $B_1,B_2,\dots,B_m$ 
			have size $r$.
			\item $\sum_{i=1}^t{|B_i|}=k-1-s+\delta-1$
		\end{itemize}
	\end{enumerate}
	In particular, when $t=0$, we only do the first step. It is easy to justify the above method by dividing into two cases $t=m$ and $t=m+1$. Note that $|A_i|-|B_i|\geq \delta-1$ for any $i=1,2,\dots,t,L+1$. We define $$\overline{B_i}=\begin{cases}
		A_i & \text{ if } |A_i|-|B_i|=\delta-1,\\
		B_i & \text{ if } |A_i|-|B_i|\text{ }\textgreater\text{ }\delta-1.
	\end{cases}$$
	for $i=1,2,\dots,t,L+1$.
	
	Let $A\triangleq B_{L+1}\bigcup\cup_{i=1}^{t}B_i$ and $S\triangleq \overline{B_{L+1}}\bigcup\cup_{i=1}^{t}\overline{B_i}$. Since $|A|=k-1$, we have $|\mathcal{C}_{A}|\leq q^{k-1}\textless q^k$. By the definition of $\overline{B_i}$ and $S$, we have $|S|=k-1+(m+1)(\delta-1)=k-1+(\lceil\frac{k+v}{r}\rceil-1)(\delta-1)$. By the locality of $\mathcal{C}$, we have $|\mathcal{C}_S|=|\mathcal{C}_A|$.
	Thus there always exists $S\subseteq [n]$ such that $|S|\geq k-1+(\lceil\frac{k+v}{r}\rceil-1)(\delta-1) $ and $|\mathcal{C}_{S}|\textless q^k$. So we have $d\leq n-k+1-(\lceil\frac{k+v}{r}\rceil-1)(\delta-1)$ by \eqref{minimum distance}.
\end{proof}
Next we apply this bound to construct distance-optimal $(r,\delta)$-LRCs among which the code length can be up to $q-1+\delta$. They can attain the Singleton-type bound \eqref{S-L, 2} for some particular dimension $k$.
\section{Constructions via good polynomials of special forms}\label{sec:3Cons}
\subsection{Making use of some smaller set where good polynomial takes a constant}
In \cite[Construction 1]{kolosov2018optimal}, for the case of $\delta=2$, Kolosov \textit{et al.} proposed a technique to hide at most $r-1$ evaluation points from an arbitrary $(r+1)$-set where the good polynomial $g(x)$ takes a constant, which allows more flexibility on the code length $n$ compared to the RS-like LRCs first introduced in \cite{tamo2014family}. One may note that a good polynomial $g(x)$ can be a constant on some smaller subset of $\mathbb{F}_q$ whose size is less than $\deg(g(x))$. The technique in \cite{kolosov2018optimal} could be adjusted to including such smaller set, and then longer codes than that in \cite{tamo2014family} appear. 
\\
\textbf{Construction A}: Suppose that $g(x)=(x-b_1)(x-b_2)\dots(x-b_s)g_1(x)$ is a good polynomial of degree $r+\delta-1$ that takes non-zero constants on each of the pairwise disjoint sets $A_i=\{a_{i,j}|j=1,2,\dots,r+\delta-1\}\subseteq \mathbb{F}_q,i=1,2,\dots,L$, where $b_1,b_2,\dots,b_s$ are pairwise distinct.

Assume that $1\leq v:=\deg(g_1(x))=r+\delta-1-s\leq r-1$ and $r\leq k\leq Lr+s-\delta+1$.

For the vector $\boldsymbol{I}=(I_{0},I_1,\dots,I_{r-v-1};{I_{i,j},i=0,1,\dots,r-1,j\in [\xi(i)]})\in \mathbb{F}_q^k$ (see \eqref{Infromation, I, form}), we define the evaluation polynomials $S_{\boldsymbol{I}}(x)$ and $F_{\boldsymbol{I}}(x)$ by: 
\begin{equation}\label{ConsA, S(x)}
	S_{\boldsymbol{I}}(x)=I_0+I_1x+\dots+I_{r-1-v}x^{r-1-v},
\end{equation}
\begin{equation}\label{ConsB, F(x)}
	F_{\boldsymbol{I}}(x)=S_{\boldsymbol{I}}(x)g_1(x)+\sum_{i=0}^{r-1} \sum_{j=1}^{\xi(i)} I_{i,j}x^ig(x)^j.
\end{equation}
Then we define the code $\mathcal{C}$ by
$$\mathcal{C}=\{({S_{\boldsymbol{I}}(b_1),S_{\boldsymbol{I}}(b_2),\dots,S_{\boldsymbol{I}}(b_s)};F_{\boldsymbol{I}}(a_{i,j}),i{=}1,2,\dots,L,j{=}1,2,\dots,r+\delta-1)|\boldsymbol{I}\in \mathbb{F}_q^k\}.$$
The code length of $\mathcal{C}$ is $n=L(r+\delta-1)+s=(L+1)(r+\delta-1)-v$.
Now we determine some other parameters of the code.

\begin{proposition}\label{prop:ConsAB}
	The code $\mathcal{C}$ defined in \textbf{Construction A} is an $(r,\delta)$-LRC with code length $n=(L+1)(r+\delta-1)-v=L(r+\delta-1)+s$, dimension $k$ and minimum distance $d=n-k-(\lceil\frac{k+v}{r}\rceil-1)(\delta-1)+1$.
\end{proposition}
\begin{proof}
	For vector $\boldsymbol{I}\in \mathbb{F}_q^k$, let $\boldsymbol{c}_{\boldsymbol{I}}$ be its corresponding codeword of $\mathcal{C}$ in \textbf{Construction A}.
	
	Since $\deg(S_{\boldsymbol{I}}(x))\leq s-\delta$, $S_{\boldsymbol{I}}(x)$ has either $s$ or at most $s-\delta$ zeros in the $s$-set $A_{L+1}:=\{b_1,b_2,\dots,b_s\}$. Since $\deg(F_{\boldsymbol{I}}(x)|_{A_i})\allowbreak\leq r-1$ for any $i\in [L]$, $F_{\boldsymbol{I}}(x)$ has either $r+\delta-1$ or $\leq r-1$ zeros in the $(r+\delta-1)$-set $A_{i}$.	So we have $\D(\mathcal{C}|_{A_{i}})\geq \delta$ or $\mathcal{C}|_{A_{i}}=\{\boldsymbol{0}\}$ for any $i\in [L+1]$. $\mathcal{C}$ has $(r,\delta)$-locality.	
	
	Recall that we use $k'$ to denote $k+v$ for simplicity. Next we consider the dimension and minimum distance of $\mathcal{C}$.
	
	Assume that $\boldsymbol{I}\in \mathbb{F}_q^k$ is non-zero. 
	It is easy to see that $\deg(F_{\boldsymbol{I}}(x))\geq \deg(g_1(x))$. We also have the following upper bound on
	the degree of $F_{\boldsymbol{I}}(x)$.
	
	If $r\nmid k'$, then $\deg(F_{\boldsymbol{I}}(x))\leq \lfloor\frac{k'}{r}\rfloor(r+\delta-1)+(k'\mod r)-1=k'+(\lceil\frac{k'}{r}\rceil-1)(\delta-1)-1$ by \eqref{Infromation, I, form} and \eqref{ConsB, F(x)}. 
	
	If $r\mid k'$, then $\deg(F_{\boldsymbol{I}}(x))\leq (\lfloor\frac{k'}{r}\rfloor-1)(r+\delta-1)+r-1=k'+(\lceil\frac{k'}{r}\rceil-1)(\delta-1)-1$ by \eqref{Infromation, I, form} and \eqref{ConsB, F(x)}.
	
	Hence, we have \begin{equation}\label{F(x) degree upper bound}
		v=\deg(g_1(x))\leq \deg(F_{\boldsymbol{I}}(x))\leq k'+(\lceil\frac{k'}{r}\rceil-1)(\delta-1)-1.
	\end{equation}
	
	To determine the lower bound of $\Wt(\boldsymbol{c}_{\boldsymbol{I}})$, we consider the following two cases.
	\begin{enumerate}[-]
		\item 
		If $S_{\boldsymbol{I}}(b_1)=S_{\boldsymbol{I}}(b_2)=\dots=S_{\boldsymbol{I}}(b_s)=0$, then $S_{\boldsymbol{I}}(x)=0$ since $\deg(S_{\boldsymbol{I}}(x))\leq s-\delta$, then we have $g(x)|F_{\boldsymbol{I}}(x)$. 
		Note that $\frac{F_{\boldsymbol{I}}(x)}{g(x)}$ is a non-zero polynomial and $\deg(\frac{F_{\boldsymbol{I}}(x)}{g(x)})\leq k'+(\lceil\frac{k'}{r}\rceil-1)(\delta-1)-1-(r+\delta-1)$. Since $g(a)\neq 0$ for any $a\in \cup_{i=1}^{L}A_i$, there are at most $k'+(\lceil\frac{k'}{r}\rceil-1)(\delta-1)-1-(r+\delta-1)$ roots of $F_{\boldsymbol{I}}(x)$ in the set $\cup_{i=1}^{L}A_i$.
		
		Thus, we have \begin{align*}
			\Wt(\boldsymbol{c}_{\boldsymbol{I}})
			&\geq L(r+\delta-1)-\bigg(k'+(\lceil\frac{k'}{r}\rceil-1)(\delta-1)-1-(r+\delta-1)\bigg)\\
			&=n-s+1-k-v-(\lceil\frac{k'}{r}\rceil-1)(\delta-1)+r+\delta-1\\
			&=n-k-(\lceil\frac{k'}{r}\rceil-1)(\delta-1)+1.\\
		\end{align*}
		\item
		If $\big(S_{\boldsymbol{I}}(b_1),S_{\boldsymbol{I}}(b_2),\dots,S_{\boldsymbol{I}}(b_s)\big)\neq \boldsymbol{0}\in \mathbb{F}_q^{s}$, since $\deg(S_{\boldsymbol{I}}(x))\leq s-\delta$, we assume that $S_{\boldsymbol{I}}(b)=0$ for any $b\in\{b_{i_1},b_{i_2},\dots,b_{i_e}\}\subseteq A_{L+1}$ and $S_{\boldsymbol{I}}(b)\neq 0$ for any $b\in A_{L+1}\backslash\{b_{i_1},b_{i_2},\dots,b_{i_e}\}$, where $e\leq s-\delta$. We have $\prod_{j=1}^{e}(x-b_{i_j})|S_{\boldsymbol{I}}(x)$ since $b_{i_1},b_{i_2},\dots,b_{i_e}$ are distinct roots of $S_{\boldsymbol{I}}(x)$. Then by \eqref{ConsB, F(x)}, we have $\big(g_1(x)\prod_{j=1}^{e}(x-b_{i_j})\big)\big|F_{\boldsymbol{I}}(x)$.
		Note that $\frac{F_{\boldsymbol{I}}(x)}{g_1(x)\prod_{j=1}^{e}(x-b_{i_j})}$ is a non-zero polynomial whose degree is at most $k'+(\lceil\frac{k'}{r}\rceil-1)(\delta-1)-1-(e+v)$. There are at most $k'+(\lceil\frac{k'}{r}\rceil-1)(\delta-1)-1-(v+e)$ roots of $F_{\boldsymbol{I}}(x)$ in the set $\cup_{i=1}^{L}A_i$ since $g_1(x)\prod_{j=1}^{e}(x-b_{i_j})$ has no root in the
		set $\cup_{i=1}^{L}A_i$.
	\end{enumerate}  
	Thus, we have\begin{align*}
		\Wt(\boldsymbol{c}_{\boldsymbol{I}})
		&\geq L(r+\delta-1)-\bigg(k'+(\lceil\frac{k'}{r}\rceil-1)(\delta-1)-1-(e+v)\bigg)+s-e\\
		&=L(r+\delta-1)+s+1-k'-(\lceil\frac{k'}{r}\rceil-1)(\delta-1)+v\\
		&=n-k-(\lceil\frac{k'}{r}\rceil-1)(\delta-1)+1.
	\end{align*}
	Since $k\leq Lr+s-\delta+1$, we have $\Wt(\boldsymbol{c}_{\boldsymbol{I}})\geq n-k-(\lceil\frac{k'}{r}\rceil-1)(\delta-1)+1\geq \delta$. The linear map $\boldsymbol{I}\mapsto \boldsymbol{c}_{\boldsymbol{I}}$ is an injective since its kernel only contains $\boldsymbol{0}$. Hence, the dimension of $\mathcal{C}$ is $k$ and $d= n-k-(\lceil\frac{k'}{r}\rceil-1)(\delta-1)+1$ by Proposition \ref{prop:MiniDistanceUpperBound}.
\end{proof}
We summarize $\textbf{Construction A}$ in the following theorem.
\begin{theorem}\label{thm:ConsAB}
	Suppose $2\leq \delta\leq s\leq r+\delta-2$, $2\leq r\leq k\leq Lr+s-\delta+1$ and $g(x)=(x-b_1)(x-b_2)\dots(x-b_s)g_1(x)\in \mathbb{F}_q[x]$ is a good polynomial of degree $r+\delta-1$ that takes non-zero constants on $L$ disjoint $(r+\delta-1)$-sets, where $b_1,b_2,\dots,b_s$ are distinct, then there explicitly exists a distance-optimal $(r,\delta)$-LRC with parameters $$[n=L(r+\delta-1)+s,k,d=n-k-(\lceil\frac{k+v}{r}\rceil-1)(\delta-1)+1]_q,$$
	where $v=r+\delta-1-s$.
\end{theorem}
\begin{remark}\label{rem:ConsBrem}
	(1) One may note that if $v=0$ and $g_1(x)=1$ in \textbf{Construction A}, the code $\mathcal{C}$ defined in it is the same as that defined in \cite[Construction 1]{tamo2014family} and \cite[Construction 8]{tamo2014family}, which means that \textbf{Construction A} is indeed a generalization. We omit this case only because it will not produce $(r,\delta)$-LRCs with new parameters.
	
	(2) In \textbf{Construction A} and Theorem \ref{thm:ConsAB}, if $1\leq (k \mod r)\leq r-v$, then $\mathcal{C}$'s minimum distance  $$d=n-k-(\lceil\frac{k+v}{r}\rceil-1)(\delta-1)+1=n-k-(\lceil\frac{k}{r}\rceil-1)(\delta-1)+1$$
	attains the Singleton-type bound \eqref{S-L, 2}.
	Note that $(n\mod (r+\delta-1))=s$ and $r+\delta-1\textless s+r\textless 2(r+\delta-1)$. We have $$\big((n+r)\mod (r+\delta-1)\big)=\big((s+r)\mod (r+\delta-1)\big)=s+r-(r+\delta-1)=r-v.$$ Thus, if $1\leq (k \mod r)\leq \big((n+r)\mod (r+\delta-1)\big)$, then $\mathcal{C}$ is a Singleton-optimal $(r,\delta)$-LRC.
	
	(3) For a polynomial $f(x)\in \mathbb{F}_q[x]$ of degree $m\geq 3$, if $G(f)\geq 1$ and $N_{f}(s)\geq 1$ for some $2\leq s\leq m-1$, then there exists $t_0\in\mathbb{F}$ such that $g(x)=f(x)-t_0$ has exactly $s$ distinct roots $b_1,b_2,\dots,b_s$ in $\mathbb{F}$. It is easy to verify that $G(f)=G(g)$ and $g(x)=(x-b_1)(x-b_2)\dots(x-b_s)g_1(x)$, which satisfies the condition of $\textbf{Construction A}$. The code length can be increased by at most $s$ compared to known RS-like LRCs. It is meaningful to consider the values $N_f(2),N_f(3),\dots,N_{f}(m-1)$ when $G(f)$ is fixed. For example, by \cite[Table 3]{chen2021good}, when $q\equiv 3\mod 8$, for the polynomial $f(x)=x^4+a_1x^2$ over $\mathbb{F}_q$, if $a_1\neq 0$ is a square, then $G(f)=(q-3)/8,N_f(2)=N_f(3)=0$, if $a_1\neq 0$ is a non-square, then $G(f)=(q-3)/8$, $N_f(3)=1$, if $a_1=0$, then $G(f)=0$. Then the code length can be increased by $3$ when $a_1$ is a non-square.
\end{remark}
Now we review some other known good polynomials that satisfy conditions of \textbf{Construction A}.
\begin{lemma}[{{\cite[Theorem 3.3]{tamo2014family}}}]\label{lem:ConsBlem1}
	Let $l,s,m$ be integers such that $l\mid s$, $p^l \equiv 1\mod m$, and $p$ is a prime. Let $H$ be an additive subgroup
	of the field $\mathbb{F}_{p^s}$ that is closed under the multiplication by the
	field $\mathbb{F}_{p^l}$ , and let $\alpha_1,\dots,\alpha_m$ be the $m$-th degree roots of unity
	in $F_{p^s}$. Then for any $b\in \mathbb{F}_{p^s}$, the polynomial $g(x)=\prod_{i=1}^{m}\prod_{h\in H}(x+h+\alpha_i)$ is constant on the union of cosets of $H$, $\cup_{1\leq i\leq m}H + b\alpha_i$, and
	the size of this union satisfies
	$$|\cup_{1\leq i\leq m}H + b\alpha_i|=\begin{cases}
		|H|&\text{if} \quad b\in H\\
		m|H|&\text{if} \quad b\notin H\\
	\end{cases}$$
\end{lemma}
The following corollary is an immediate consequence of Lemma \ref{lem:ConsBlem1} and Remark \ref{rem:ConsBrem} (3). 
\begin{corollary}\label{cor:ConsBcor1}
	In Lemma \ref{lem:ConsBlem1}, if $m\geq 2$ and $2\leq |H|\textless q$, then we have $G(g)=\frac{q-|H|}{m{|H|}}$ and $N_{g}(|H|)=1$. For any $m|H|-|H|+1\leq r\leq m|H|-1$ and $\delta=m|H|-r+1$, we have a distance-optimal $(r,\delta)$-LRC whose length is $n=q$ and minimum distance is $n-k-(\lceil\frac{k+m|H|-|H|}{r}\rceil-1)(\delta-1)+1$ by Theorem \ref{thm:ConsAB}, where $k\in [r,r\frac{q-|H|}{m{|H|}}+|H|-\delta+1]$ is the dimension of this code.
\end{corollary}
The Dickson polynomials (of the first kind) of degree $m$ for some $a\in \mathbb{F}_q$ is defined by $$D_m(x,a)=\sum\limits_{i=0}^{\lfloor m/2\rfloor}\frac{m}{m-i}\binom{m-i}{i}(-a)^ix^{m-2i}.$$
%
\begin{lemma}[{{\cite[Theorem 19]{chen2022function}}}]\label{lem:ConsBlem2}
	If $f(x)=D_m(x, a)-D_m(0, a)$ for $a\in \mathbb{F}_q^*$, $m\textgreater 2$ with $\gcd(m, q)=1$ and
	$q\equiv \pm 1 (\mod m)$, then
	\begin{equation}\label{Dickson, 2, equation}
		G(f)=\begin{cases}
			\lfloor\frac{q-3}{2m}\rfloor&\mbox{if}\quad q \quad\mbox{is odd and}\quad q\equiv \eta(a)\equiv 1\mod m,\\
			
			\lfloor\frac{q+1}{2m}\rfloor&\mbox{if}\quad q \quad\mbox{is odd and}\quad q\equiv \eta(a)\equiv −1 \mod m,\\
			
			\lfloor\frac{q}{2m}\rfloor&\mbox{otherwise},
		\end{cases}
	\end{equation}
	where $\eta$ is the quadratic character of $\mathbb{F}_q^*$, i.e.,
	$$\eta(c)=
	\begin{cases}
		0&if\quad c = 0,\\
		1&if \quad c\in \mathbb{F}_q^*\quad \mbox{is a square},\\
		-1&if \quad c\in \mathbb{F}_q^*\quad \mbox{is a non-square}.
	\end{cases}
	$$
\end{lemma}
\smallskip
If $2\nmid m$, $\gcd(m, q) = 1$ and $\omega\in\mathbb{F}_{q^2}$
is a primitive $m$-th root of unity, then
\begin{equation}\label{Dickson, factor}
	D_m(x,a)-D_m(x,0)=x\prod_{i=1}^{(m-1)/2}(x^2+a(\omega^i-\omega^{-i})^2)
\end{equation}
Assume that $q$ is even and $q\equiv -1 \mod m$, then $\omega^i-\omega^{-i}\neq 0$ and $$((\omega^i-\omega^{-i})^2)^{q-1}{=}\frac{((\omega^i+\omega^{-i})^2)^{q}}{(\omega^i+\omega^{-i})^2}{=}\frac{\omega^{2qi}+\omega^{-2qi}}{\omega^{2i}+\omega^{-2i}}{=}\frac{\omega^{2i(q+1)}/\omega^{2i}+\omega^{-2i(q+1)}/\omega^{-2i}}{\omega^{2i}+\omega^{-2i}}=1,$$
where the last
equation is due to the facts that $\omega$ is the $m$-th primitive roots and $m|(q+1)$.
If $\omega^i-\omega^{-i}=\omega^j-\omega^{-j}$, then we have $(\omega^{i+j}-1)(\omega^i-\omega^j)=0$. 
Hence,  $(\omega^1-\omega^{-1})^2,(\omega^2-\omega^{-2})^2,\dots,(\omega^{(m-1)/2}-\omega^{-(m-1)/2})^2$ are pairwise distinct non-zero elements and are all in $\mathbb{F}_q$. Thus, $f(x)=D_m(x,a)-D_m(x,0)$ has exactly $\frac{m+1}{2}$ distinct roots in $\mathbb{F}_q$ by $\eqref{Dickson, factor}$. We have the following corollary by Lemma \ref{lem:ConsBlem2} and \textbf{Construction A}.
\begin{corollary}\label{cor:ConsBcor2}
	If $m\textgreater 2$, $q$ is even and $q\equiv -1\mod m$, then for any $\frac{m+1}{2}\leq r\leq m-1$ and $\delta=m+1-r$, we have a distance-optimal $(r,\delta)$-LRC whose length is $n=m\lfloor\frac{q}{2m}\rfloor+\frac{m+1}{2}$ and minimum distance is $n-k-(\lceil\frac{k+(m-1)/2}{r}\rceil)(\delta-1)+1$, where $k\in [r,r\lfloor\frac{q}{2m}\rfloor+\frac{m+1}{2}-\delta+1]$ is the dimension of this code.
\end{corollary}
There are also some other explicit good polynomials satisfying the condition of \textbf{Construction A}, such as \cite[Theorem 9]{chen2021good}. We have an explicit example as follows:
\begin{example}\label{exam:ConsBexam}
		By \cite[Theorem 9]{chen2021good}, assume that $g(x) = x^6+a_4x^4+a_2x^2$ is a polynomial over $\mathbb{F}_q$, where $q\equiv 3 \mod 6$ is a prime power, if $a_4^2=a_2\neq 0$ and $a_4$ is a square in $\mathbb{F}_q$ , then
		$G(g)=\frac{q-3}{6}$. 
		
		Since $3|q$, we have  $$g(x)=x^6+a_4x^4+a_2x^2=x^2(x^4+2\frac{a_4}{2}x^2+(\frac{a_4}{2})^2)=x^2(x^2+\frac{a_4}{2})^2=x^2(x^2-a_4)^2.$$
		Since $a_4$ is a non-zero square in $\mathbb{F}_q$, $g(x)$ is a good polynomial satisfies the condition of \textbf{Construction A}, in which the code length of $\mathcal{C}$ can be up to $6\frac{q-3}{6}+3=q$.
		
		Let $u$ be a root of the polynomial $x^3 + 2x + 1$ over $\mathbb{F}_3$. 
		We have a good polynomial $x^6+u^2x^4+u^4x^2=x^2(x-u)^2(x+u)^2\in \mathbb{F}_{27}$ that is constant on each of the following four sets. 
		$$A_1=\{u^2,u^4,u^{10},u^{15},u^{17},u^{23}\},g(A_1)=u^6,$$
		$$A_2=\{u^7,u^{11},u^{12},u^{20},u^{24},u^{25}\},g(A_2)=u^8,$$
		$$A_3=\{u^5,u^6,u^8,u^{18},u^{19},u^{21}\},g(A_3)=u^{12},$$
		$$A_4=\{u^3,u^9,u^{13},u^{16},u^{22},u^{26}\},g(A_4)=u^{24}.$$
		
		In \textbf{Construction A}, let $r=5,\delta=2,b_1=0,b_2=u,b_3=u^{14},g_1(x)=x(x-u)(x-u^{14})$ and $k=12$. 
		
		For the vector $\boldsymbol{I}{=}\{I_0,I_1,I_{0,1},I_{1,1},I_{2,1},I_{3,1},I_{4,1},I_{0,2},I_{1,2},I_{2,2},I_{3,2},I_{4,2}\}$, the evaluation polynomials are defined as follows.
		\begin{equation*}
			S_{\boldsymbol{I}}(x)=I_0+I_1x,\quad
			F_{\boldsymbol{I}}(x)=S_{\boldsymbol{I}}(x)g_1(x)+\sum_{i=0}^{4} \sum_{j=1}^{2} I_{i,j}x^ig(x)^j.
		\end{equation*}
	Then we have a generator matrix of this $(5,2)$-LRC as follows with the help of the computer algebra system Magma calculator \cite{bosma1997magma},
		\setlength{\arraycolsep}{1.5pt}
		\begin{align*}
			&\left[\begin{array}{ccc|cccccc|cccccc}
				1&1&1&u^{16}&u^{16}&u^{16}&u^3&u^3&u^3&u^4&u^4&u^4&u^{17}&u^{17}&u^{17}\\
				0&u&u^{14}&u^{18}&u^{20}&1&u^{18}&u^{20}&1&u^{11}&u^{15}&u^{16}&u^{11}&u^{15}&u^{16}\\
				0&0&0&u^{6}&u^{6}&u^{6}&u^{6}&u^{6}&u^{6}&u^{8}&u^{8}&u^{8}&u^{8}&u^{8}&u^{8}\\
				0&0&0&u^{8}&u^{10}&u^{16}&u^{21}&u^{23}&u^{3}&u^{15}&u^{19}&u^{20}&u^{2}&u^{6}&u^{7}\\
				0&0&0&u^{10}&u^{14}&1&u^{10}&u^{14}&1&u^{22}&u^{4}&u^{6}&u^{22}&u^{4}&u^{6}\\
				0&0&0&u^{12}&u^{18}&u^{10}&u^{25}&u^{5}&u^{23}&u^{3}&u^{15}&u^{18}&u^{16}&u^{2}&u^{5}\\
				0&0&0&u^{14}&u^{22}&u^{20}&u^{14}&u^{22}&u^{20}&u^{10}&1&u^{4}&u^{10}&1&u^{4}\\
				0&0&0&u^{12}&u^{12}&u^{12}&u^{12}&u^{12}&u^{12}&u^{16}&u^{16}&u^{16}&u^{16}&u^{16}&u^{16}\\
				0&0&0&u^{14}&u^{16}&u^{22}&u&u^{3}&u^{9}&u^{23}&u&u^{2}&u^{10}&u^{14}&u^{15}\\
				0&0&0&u^{16}&u^{20}&u^{6}&u^{16}&u^{20}&u^{6}&u^{4}&u^{12}&u^{14}&u^{4}&u^{12}&u^{14}\\
				0&0&0&u^{18}&u^{24}&u^{16}&u^{5}&u^{11}&u^{3}&u^{11}&u^{23}&1&u^{24}&u^{10}&2\\
				0&0&0&u^{20}&u^{2}&1&u^{20}&u^{2}&1&u^{18}&u^{8}&u^{12}&u^{18}&u^{8}&u^{12}\\
			\end{array}\sim\right.\\
			&\quad\quad\quad\quad\quad\quad\quad\quad\quad\quad\quad\quad\quad\quad\left. \sim \begin{array}{cccccc|cccccc}
				u^6&u^{19}&u^6&u^{19}&u^6&u^{19}&u^{12}&u^{25}&u^{12}&u^{25}&u^{12}&u^{25}\\
				u^{11}&u^{25}&u^{14}&u^{11}&u^{25}&u^{14}&u^{15}&u^{8}&u^{25}&u^{15}&u^{8}&u^{25}\\
				u^{12}&u^{12}&u^{12}&u^{12}&u^{12}&u^{12}&u^{24}&u^{24}&u^{24}&u^{24}&u^{24}&u^{24}\\
				u^{17}&u^{18}&u^{20}&u^{4}&u^{5}&u^{7}&u&u^{7}&u^{11}&u^{14}&u^{20}&u^{24}\\
				u^{22}&u^{24}&u^{2}&u^{22}&u^{24}&u^{2}&u^{4}&u^{16}&u^{24}&u^{4}&u^{16}&u^{24}\\
				u&u^{4}&u^{10}&u^{14}&u^{17}&u^{23}&u^{7}&u^{25}&u^{11}&u^{20}&u^{12}&u^{24}\\
				u^{6}&u^{10}&u^{18}&u^{6}&u^{10}&u^{18}&u^{10}&u^{8}&u^{24}&u^{10}&u^{8}&u^{24}\\
				u^{24}&u^{24}&u^{24}&u^{24}&u^{24}&u^{24}&u^{22}&u^{22}&u^{22}&u^{22}&u^{22}&u^{22}\\
				u^{3}&u^{4}&u^{6}&u^{16}&u^{17}&u^{19}&u^{25}&u^{5}&u^{9}&u^{12}&u^{18}&u^{22}\\
				u^{8}&u^{10}&u^{14}&u^{8}&u^{10}&u^{14}&u^{2}&u^{14}&u^{22}&u^{2}&u^{14}&u^{22}\\
				2&u^{16}&u^{22}&1&u^{3}&u^{9}&u^{5}&u^{23}&u^{9}&u^{18}&u^{10}&u^{22}\\
				u^{18}&u^{22}&u^{4}&u^{18}&u^{22}&u^{4}&u^{8}&u^{6}&u^{22}&u^{8}&u^{6}&u^{22}\\
			\end{array}\right]
		\end{align*}
		The minimum distance of this LRC is $14$. It is a Singleton-optimal LRC with locality $6$ whose code length is equal to the field size $27$.
	\end{example}
	\subsection{Beyond the field size}
	In the proof of Proposition \ref{prop:ConsAB}, we find that if $v=r-1$, i.e., $s=\delta$, then the evaluation points for $S_{\boldsymbol{I}}(x)$ is unnecessary since $S_{\boldsymbol{I}}(x)$ is a constant polynomial. Thus, we have the following \textbf{Construction B}. It is worth noting that the code length of $\mathcal{C}$ in it can exceed the field size $q$ for some good polynomials.
	\\ \textbf{Construction B}: Suppose that $g(x)=g_2(x)g_1(x)\in \mathbb{F}_q[x]$ is a good polynomial of degree $(r+\delta-1)$ that is constant on each of the pairwise disjoint sets $A_i=\{a_{i,j}|j=1,2,\dots,r+\delta-1\}\subseteq \mathbb{F}_q,i=1,2,\dots,L$ and $g(a)\neq 0$ for any $a\in \cup_{i=1}^{L}A_i$, where $\deg(g_2(x))=s=\delta$. 
	
	Let $v=\deg(g_1(x))=r-1$. Assume that $r\leq k\leq Lr+1$.
	
	For the vector $\boldsymbol{I}=(I_{0};{I_{i,j},i=0,1,\dots,r-1,j\in [\xi(i)]})\in \mathbb{F}_{q}^k$ (see \eqref{Infromation, I, form}), we define the evaluation polynomial $F_{\boldsymbol{I}}(x)$ by: 
	\begin{equation}\label{ConsC, F(x)}
		F_{\boldsymbol{I}}(x)=I_0g_1(x)+\sum_{i=0}^{r-1} \sum_{j=1}^{\xi(i)} I_{i,j}x^ig(x)^j.
	\end{equation}
	
	Then we define the code $\mathcal{C}$ by:
	$$\mathcal{C}=\{(\underbrace{I_0,I_0,\dots,I_0}_{\delta},F_{\boldsymbol{I}}(a_{i,j}),i=1,2,\dots,L,j=1,2,\dots,r+\delta-1)|\boldsymbol{I}\in \mathbb{F}_q^k\}.$$
	Its code length $n=L(r+\delta-1)+\delta=(L+1)(r+\delta-1)-(r-1)$. Now we determine the other parameters of this code.
	\begin{proposition}\label{prop:ConsC}
		The code $\mathcal{C}$ defined in \textbf{Construction B} is an $(r,\delta)$-LRC with code length $n=L(r+\delta-1)+\delta$, dimension $k$ and  minimum distance $d=n-k-(\lceil\frac{k+r-1}{r}\rceil-1)(\delta-1)+1$.
	\end{proposition}
	%
	%
	\begin{proof}
		First we prove that $\mathcal{C}$ has $(r,\delta)$-locality.
		We denote the coordinate positions of the $\delta$ copies of $I_0$ by $b_1^{(1)},b_1^{(2)},\dots,b_{1}^{(\delta)}$. Let $A_{L+1}=\{b_1^{(1)},b_1^{(2)},\dots,b_{1}^{(\delta)}\}$.
		For the other $L(r+\delta-1)$ coordinate positions, we use the evaluation points to denote their corresponding coordinate positions by a little abuse of notation.
		
		It is easy to verify that $\mathcal{C}$ is a linear code. We now show that $\mathcal{C}$ has $(r,\delta)$-locality. We have $\D(\mathcal{C}_{A_{L+1}})\geq \delta$ or $\mathcal{C}_{A_{L+1}}=\{\boldsymbol{0}\}$ since the values at these $\delta$ coordinate positions are all the same, i.e., $I_0$. Since $\deg(F_{\boldsymbol{I}}(x)|_{A_i})\allowbreak\leq r-1$ for any $i\in [L]$, $F_{\boldsymbol{I}}(x)$ has either $r+\delta-1$ or $\leq r-1$ zeros in the $(r+\delta-1)$-set $A_{i}$. 
		%
		So we have $\D(\mathcal{C}|_{A_{i}})\geq \delta$ or $\mathcal{C}|_{A_{i}}=\{\boldsymbol{0}\}$ for any $i\in [L+1]$. $\mathcal{C}$ has $(r,\delta)$-locality.
		
		Recall that $k'=k+v=k+r-1$. Next we prove that the dimension of $\mathcal{C}$ is $k$ and minimum distance is $d=n-k-\lceil\frac{k+r-1}{r}\rceil(\delta-1)$.

		Assume that $\boldsymbol{I}\in \mathbb{F}_q^k$ is a non-zero vector and $\boldsymbol{c}_{\boldsymbol{I}}$ is its corresponding codeword. 
		We have $$r-1\leq \deg(F_{\boldsymbol{I}}(x))\leq k'+\lceil\frac{k+r-1}{r}-1\rceil(\delta-1)-1$$ by the same argument of \eqref{F(x) degree upper bound} in the proof of Proposition \ref{prop:ConsAB}. To determine the lower bound of $\Wt(\boldsymbol{c}_{\boldsymbol{I}})$, we consider the following two cases.
		\begin{enumerate}[-]
			\item
			If $I_0=0$, then $g(x)|F_{\boldsymbol{I}}(x)$, $F_{\boldsymbol{I}}(x)=g(x)\frac{F_{\boldsymbol{I}}(x)}{g(x)}$. Note that $g(a)\neq 0$ for any $a\in \cup_{i=1}^{L}A_i$. $F_{\boldsymbol{I}}(x)$ has at most $k'+(\lceil\frac{k+r-1}{r}\rceil-1)(\delta-1)-1-(r+\delta-1)$ roots in $\cup_{i=1}^{L}A_i$ since $0\leq \deg(\frac{F_{\boldsymbol{I}}(x)}{g(x)})\leq k'+(\lceil\frac{k+r-1}{r}\rceil-1)(\delta-1)-1-(r+\delta-1)$. Hence, we have 
			\begin{align*}
				\Wt(\boldsymbol{c}_{\boldsymbol{I}})&\geq  n-\delta-\bigg(k+(\lceil\frac{k+r-1}{r}\rceil-1)(\delta-1)-1+v-(r+\delta-1)\bigg)\\&=n-k-(\lceil\frac{k+r-1}{r}\rceil-1)(\delta-1)+1.
			\end{align*}
			
			\item
			If $I_0\neq 0$, then $g_1(x)|F_{\boldsymbol{I}}(x)$, $F_{\boldsymbol{I}}(x)=g_1(x)\frac{F_{\boldsymbol{I}}(x)}{g_1(x)}$. Note that $g_1(a)\neq 0$ for any $a\in \cup_{i=1}^{L}A_i$. $F_{\boldsymbol{I}}(x)$ has at most $k'+(\lceil\frac{k+r-1}{r}\rceil-1)(\delta-1)-1-(r-1)$ roots in $\cup_{i=1}^{L}A_i$ since $0\leq \deg(\frac{F_{\boldsymbol{I}}(x)}{g_1(x)}) \leq k'+(\lceil\frac{k+r-1}{r}\rceil-1)(\delta-1)-1-(r-1)$. Hence, we have 
			\begin{align*}
				\Wt(\boldsymbol{c}_{\boldsymbol{I}})&\geq  n-\delta-\bigg(k+(\lceil\frac{k+r-1}{r}\rceil-1)(\delta-1)-1+v-(r-1)\bigg)+\delta\\&=n-k-(\lceil\frac{k+r-1}{r}\rceil-1)(\delta-1)+1.
			\end{align*}
		\end{enumerate}
		Since $k\leq Lr+1$, we have $\Wt(\boldsymbol{c}_{\boldsymbol{I}})\geq n-k-(\lceil\frac{k+r-1}{r}\rceil-1)(\delta-1)+1\geq \delta$. It is easy to see that the map $\boldsymbol{I}\mapsto \boldsymbol{c}_{\boldsymbol{I}}$ is linear. This map is also an injective since its kernel only contains $\boldsymbol{0}$. Hence, $\mathcal{C}$ has dimension $k$ and minimum distance $d= n-k-(\lceil\frac{k+r-1}{r}\rceil-1)(\delta-1)+1$ by Proposition \ref{prop:MiniDistanceUpperBound}.
	\end{proof}
	Now we can summarize \textbf{Construction B} by the following theorem.
	\begin{theorem}\label{thm:ConsC}
		Suppose $2\leq \delta$, $2\leq r\leq k\leq Lr+1$ and $g(x)=g_2(x)g_1(x)\in \mathbb{F}_q[x]$ is a good polynomial of degree $r+\delta-1$ that takes non-zero constants on $L$ disjoint $(r+\delta-1)$-sets, where $g_1(x)$ and $g_2(x)$ has degrees $r-1$ and $\delta$, respectively, then there explicitly exists a distance-optimal $(r,\delta)$-LRC with parameters $$[n=L(r+\delta-1)+\delta,k,d=n-k-(\lceil\frac{k+r-1}{r}\rceil-1)(\delta-1)+1]_q.$$
	\end{theorem}
	Let us review some known good polynomials that satisfy the condition of \textbf{Construction B}.
	\begin{lemma}[{{\cite[Proposition 3.2]{tamo2014family}}}]\label{lem:ConsClem1}
		If $H$ is a subgroup of $\mathbb{F}_q^*$ or $\mathbb{F}_q^+$, the polynomial $g(x)=\prod_{h\in H}(x-h)$ is constant on each coset of $H$.
	\end{lemma}
	Let $H$ be a subgroup of $\mathbb{F}_q^*$. For any $\alpha\in H$, $\alpha^{|H|}=1$ and $\prod_{h\in H}(\alpha-h)=0$. Since $\prod_{h\in H}(x-h),x^{|H|}$ are both monic and of degree $|H|$, we have $1+\prod_{h\in H}(x-h)=x^{|H|}$, which means that $x^{|H|}$ takes non-zero constants on the cosets of $H$. Thus, we have the following result.
	 
	\begin{corollary}\label{cor:ConsCcor1}
		If $r\geq 2$ and $\delta\geq 2$, then $g(x)=x^{r+\delta-1}=x^{\delta}\cdot x^{r-1}$ is a good polynomial with $G(g)=\frac{q-1}{r+\delta-1}$ in the field $\mathbb{F}_q$, where $(r+\delta-1)|(q-1)$. We have a distance-optimal $(r,\delta)$-LRC whose length is up to $n=q-1+\delta$ by \textbf{Construction B}. Its minimum distance is $n-k-(\lceil\frac{k+r-1}{r}\rceil-1)(\delta-1)+1$, where $k\in [r,(q-1)r/(r+\delta-1)+1]$ is the dimension of this $(r,\delta)$-LRC.
	\end{corollary}
	
	\begin{lemma}[{{\cite[Theorem 3]{chen2021good}}}]\label{lem:ConsClem2}
		Consider a cubic polynomial $f(x)=x^3+a_1x$ over $\mathbb{F}_q$. If $q\equiv 5\mod 6$, then
		$G(f)\leq \frac{q+1}{6}$,
		and the equality holds if and only if $a_1$ is a non-zero square in $\mathbb{F}_q$.
	\end{lemma}
	By Lemma \ref{lem:ConsClem2}, there are exactly $\frac{q+1}{6}$ pairwise disjoint $3$-subsets $A_i\subseteq \mathbb{F}_q$, $i\in [\frac{q+1}{6}]$, such that $f(x)=x^3+a_1x$ is constant on each of them, where $q\equiv 5\mod 6$ and $a_1$ is a non-zero square. Let $x_0\in \mathbb{F}_q\backslash \cup_{i=1}^{(q+1)/6}A_i$. Then $g(x)=f(x)-f(x_0)$ is a good polynomial with $G(g)=G(f)$ and is constant on $A_i$, $i\in [\frac{q+1}{6}]$.  $g(x)=(x-x_0)\frac{g(x)}{x-x_0}$ has no root in $\cup_{i=1}^{(q+1)/6}A_i$ since $x_0\not\in \cup_{i=1}^{(q+1)/6}A_i$. By \textbf{Construction B}, we have the following result.
	\begin{corollary}\label{cor:ConsCcor2}
		If $q\equiv 5 \mod 6$, then we have a distance-optimal $(r,2)$-LRC by Lemma \ref{lem:ConsClem2} and \textbf{Construction B}. Its length is $n=\frac{q+1}{2}+2$ and minimum distance is $n-k-\lceil\frac{k+1}{2}\rceil+2$, where $k\in [2,(q+1)/3+1]$ is the dimension of this code.
	\end{corollary}
	We have an explicit example by Corollary \ref{cor:ConsCcor1}.
	\begin{example}\label{exam:ConsCexam}
		If $q=17, g_1(x)=x, g_2(x)=x^3, r=2, \delta=3, v=1, k=7$, then for the information vector $\boldsymbol{I}=(I_0;I_{0, 1}, I_{1, 1}, I_{0, 2}, I_{1, 2}, I_{0, 3}, I_{1, 3})$, 
		the evaluation polynomial is:
		$$F_{\boldsymbol{I}}(x)=I_0g_1(x)+\sum_{i=0}^{1} \sum_{j=1}^{3} I_{i, j}x^ig(x)^j. $$
		Then the code $\mathcal{C}$ has the following generator matrix:
		\setlength{\arraycolsep}{4pt}
		\begin{equation*}
			\left[\begin{array}{ccc|cccc|cccc|cccc|cccc}
				1&1&1&1&4&13&16&6&7&10&11&3&5&12&14&2&8&9&15\\
				0&0&0&1&1&1&1&4&4&4&4&13&13&13&13&16&16&16&16\\
				0&0&0&1&4&13&16&7&11&6&10&5&14&3&12&15&9&8&2\\
				0&0&0&1&1&1&1&16&16&16&16&16&16&16&16&1&1&1&1\\
				0&0&0&1&4&13&16&11&10&7&6&14&12&5&3&2&8&9&15\\
				0&0&0&1&1&1&1&13&13&13&13&4&4&4&4&16&16&16&16\\
				0&0&0&1&4&13&16&10&6&11&7&12&3&14&5&15&9&8&2\\
			\end{array}\right]
		\end{equation*}
		It is a Singleton-optimal $(2, 3)$-LRC with parameters $[19, 7, 7]_{17}$.  
		
		If $q=49, g(x)=x^{24}$, $g_1(x)=x^{5},g_2(x)=x^{19}$,$r=6,\delta=19,v=5,k=7,$ then we have a  Singleton-optimal $(6,19)$-LRC with parameters $[67,7,43]_{49}$, which is longer than that constructed via elliptic curves whose code length is at most $63$ in $\mathbb{F}_{49}$.
	\end{example}
	\section{Conclusions}\label{sec:4Conc}
	In this paper, by generalizing the RS-like LRCs in \cite{tamo2014family} and \cite{kolosov2018optimal}, we construct new $(r,\delta)$-LRCs via some good polynomials of special forms. They attain the Singleton-type bound if the dimension $k$ satisfies $1\leq (k\mod r)\leq \big((n+r)\mod (r+\delta-1)\big)$. For other cases of dimension $k$, we also show that they are distance-optimal by proving a sharper bound under some extra constraints about the local repair groups.
	
	The code length of our constructions can be larger than that in \cite{tamo2014family} since we make use of some smaller set with size less than $r+\delta-1$. Our main innovation is to use two evaluation polynomials, which allows the roots of a good polynomial to be the extra evaluation points. In particular, when one of the evaluation polynomials is a constant polynomial, it no longer needs any evaluation point, making it possible that the code length exceeds the field size $q$. The largest code length of these codes is $q+\delta-1$, asymptotically exceeding the code length $q+2\sqrt{q}$ of that constructed via elliptic curves when $\delta$ is proportional
	to $q$. 


\section*{Declaration of competing interest}
\noindent
The authors declare that they have no known competing financial interests or personal relationships that could have appeared to influence the work reported in this paper.
\section*{Data availability}
\noindent
No data was used for the research described in the article.
\section*{Acknowledgement}
\noindent
This work is supported in part by Science and Technology Commission of Shanghai Municipality (No. 22DZ2229014).

 \bibliographystyle{elsarticle-num}
 \bibliography{GaoYuanREF.bib}
\end{document}